\documentclass[
final
]{dmtcs-episciences}

\usepackage[utf8]{inputenc}
\usepackage{subfigure}

\newtheorem{theorem}{Theorem}[section]
\newtheorem{corollary}[theorem]{Corollary}
\newtheorem{lemma}[theorem]{Lemma}

\usepackage[round]{natbib}

\author[Bellitto, Duffy and MacGillivray]{Thomas Bellitto\affiliationmark{1}\thanks{This work was initiated while the author was a postdocoral fellow at Department of Mathematics and Computer Science, University of Southern Denmark (Odense, Denmark) and supported by the Independent Research Fund Denmark under grant number DFF 7014-00037B}
  \and Christopher Duffy\affiliationmark{2}
  \and Gary MacGillivray\affiliationmark{3}}
\title[Homomorphically Full Oriented Graphs]{Homomorphically Full Oriented Graphs}
\affiliation{
  LIP6, CNRS, Sorbonne Universit\'e, Paris, France\\
  School of Mathematics and Statistics, University of Melbourne, Melbourne, Australia\\
  Department of Mathematics and Statistics,  University of Victoria, Victoria, Canada}

\begin{document}

\publicationdata{vol. 25:2}{2023}{13}{10.46298/dmtcs.9957}{2022-08-24; 2022-08-24; 2023-08-08}{2023-08-11}

\maketitle
\begin{abstract}
Homomorphically full graphs are those for which every homomorphic image is isomorphic to a subgraph. 
We extend the definition of homomorphically full to oriented graphs in two different ways.
For the first of these, we show that homomorphically full oriented graphs arise as quasi-transitive orientations of homomorphically full graphs. 
This in turn yields an efficient recognition and construction algorithms for these homomorphically full oriented graphs.
For the second one, we show that the related recognition problem is GI-hard, and that the problem of deciding if a graph admits a homomorphically full orientation is NP-complete.
In doing so we show the problem of deciding if two given oriented cliques are isomorphic is GI-complete.
\end{abstract}

\section{Introduction and Background}
Oriented graphs can be considered to arise in one of two ways: 
as a special class of digraphs defined by a restriction on the existence of directed 2-cycles, 
or  from simple graphs by assigning a direction to each edge.
The distinction between these perspectives has consequences in the study of homomorphisms of oriented graphs.
For example, each leads to a different definition of vertex colouring of oriented graphs (see the work of \cite{FHM03} and  \cite{S16}).
We use the term \emph{antisymmetric digraph} to specifically refer to digraphs with the restriction that there are no 
directed 2-cycles, and the term \emph{oriented graph} otherwise, including when the distinction between perspectives is unimportant.

For standard graph theoretic notation we refer the reader to the seminal text  by \cite{B08}.
We generally use Greek capital letters to refer to graphs and Latin capital letters to refer to antisymmetric digraphs and oriented graphs.

A \emph{homomorphism} of a graph $\Gamma$ to a graph $\Lambda$ is a function $\phi: V(\Gamma) \to V(\Lambda)$ such that
$xy \in E(\Gamma)$ implies $\phi(x)\phi(y) \in E(\Lambda)$.  
When $\phi$ is a homomorphism of $\Gamma$ to $\Lambda$ we write $\phi: \Gamma \to \Lambda$.
When the existence of a homomorphism (rather than a particular homomorphism) is of interest we write $\Gamma \to \Lambda$.
The definition of a homomorphism of a digraph $G$ to a digraph $H$ is identical.

A homomorphism of $\Gamma$ to $\Lambda$ induces a mapping of $A(\Gamma)$ to $A(\Lambda)$.
A homomorphism $\phi: \Gamma \to \Lambda$ is called \emph{complete} when both $\phi$ and the induced mapping of $A(\Gamma)$ to $A(\Lambda)$ are surjective.
If there is a complete homomorphism of $\Gamma$ to $\Lambda$, then $\Lambda$ is called a \emph{homomorphic image} of $\Gamma$.  
The corresponding definitions are the same for digraphs.

If $\Lambda$ is a subgraph of $\Gamma$, then a homomorphism $\phi$ of $\Gamma$ to $\Lambda$ is called a \emph{retraction} when $\phi(h) = h$ for every vertex $h \in V(\Lambda)$.  
If there is a retraction of $\Gamma$ to $\Lambda$, then $\Lambda$ is called a \emph{retract} of $\Gamma$. 
A retract of $\Gamma$ is necessarily an induced subgraph of $\Gamma$, but the converse is false.
The corresponding definitions for digraphs are identical.

Vertices $x$ and $y$ of a graph $G$ are called \emph{neighbourhood comparable} when $N(x) \subseteq N(y)$ or $N(y) \subseteq N(x)$.
Notice that if $N(x) \subseteq N(y)$, then there is a retraction of $G$ to $G - x$.
A graph $\Gamma$ is called \emph{homomorphically full} when every homomorphic image 
of $\Gamma$ is isomorphic to a subgraph of $\Gamma$.
The homomorphically full graphs were first characterized by \cite{B96}. % using neighbourhood comparability.
\begin{theorem}\label{thm:main} \cite{B96}
	Let $\Gamma$ be a graph.
	The following statements are equivalent:
	\begin{enumerate}[(a)]
		\item$\Gamma$ is homomorphically full.
		\item If $x$ and $y$ are non-adjacent vertices of $\Gamma$, then $x$ and $y$ are neighbourhood comparable
		%$N(x) \subseteq N(y)$ or $N(y) \subseteq N(x)$.
		\item Every homomorphic image of $\Gamma$ is isomorphic to a retract of $\Gamma$.
		\item Every homomorphic image of $\Gamma$ is isomorphic to an induced subgraph of $\Gamma$.
		\item $\Gamma$ contains neither $2K_2$ nor $P_4$ as an induced subgraph. 
		\item $\overline{\Gamma}$ is the comparability graph of an up-branching.
	\end{enumerate}
\end{theorem}

Homomorphic images of simple graphs are implicitly understood to be simple graphs.
If a homomorphic image of a simple graph could be a graph with loops, then there would be no homomorphically 
full simple graphs: a single vertex with a loop would be a homomorphic image of every graph.

Analogous to the definition of graphs, we say an oriented graph or antisymmetric digraph $G$ is \emph{homomorphically full} when every homomorphic image of $G$ is isomorphic to a subgraph of $G$.
The two perspectives on how an oriented graph $G$  arises are germane in understanding its homomorphic images and the meaning of this definition.
When $G$ is an antisymmetric digraph, a \emph{homomorphic image} of $G$ is a 
digraph $H$ (which may have directed 2-cycles) for which there is a complete homomorphism
$G \to H$.
When $G$ is an oriented graph, a \emph{homomorphic image} of $G$ is an oriented graph 
$H$ for which there is a complete homomorphism
$G \to H$.
In this case, two vertices joined by a directed path of length 2  (a \emph{$2$-dipath}) must map to different
vertices of $H$.

For a graph $\Gamma$ and non-adjacent vertices $u,v\in V(\Gamma)$, let $\Gamma_{uv}$ denote the graph produced by identifying $u$ and $v$ into a single vertex named $u_v$.
The homomorphism $\phi: \Gamma \to \Gamma_{uv}$ that  sends $u$ and $v$ to $u_v$ and fixes all other vertices is a complete homomorphism.
As such, $\Gamma_{uv}$ is a homomorphic image of $\Gamma$.
We call such a homomorphism \emph{elementary}.
Every homomorphism can be expressed as a composition of elementary homomorphisms.
These statements remain true when considered for oriented graphs and antisymmetric digraphs.
However for oriented graphs we are restricted to considering pairs of vertices  that are neither adjacent nor the ends of a $2$-dipath.
In any case, to check whether a graph, oriented graph, or antisymmetric digraph is homomorphically full it is enough to check that the target of each elementary homomorphism is isomorphic to a subgraph.

Let $H$ be the directed path on three vertices.
With Theorem \ref{thm:main} one can verify that the underlying graph (i.e., the path on three vertices) is homomorphically full.
Interpreting $H$ as an antisymmetric digraph, there is an elementary homomorphism that identifies the two ends of the directed path.
The resulting homomorphic image, the directed $2$-cycle, is not isomorphic to any subgraph of $H$ and so $H$ is not homomorphically full.
On the other hand, interpreting $H$ as an oriented graph, the vertices at the end of the path cannot be identified by an elementary homomorphism. 
In this case the only homomorphic image of $H$ is $H$ itself.
And so we conclude  $H$ is homomorphically full.

Let $G$ be a directed graph. We say that an ordered pair of non-adjacent vertices, $u,v\in V(G)$ are \emph{neighbourhood comparable} when $N^+(u) \subseteq N^+(v)$  and $N^-(u) \subseteq N^-(v)$ or $N^+(v) \subseteq N^+(u)$  and $N^-(v) \subseteq N^-(u)$.
As with graphs, we observe that if $u$ and $v$ are neighbourhood comparable, then,  we have  $G-u \cong G_{uv}$, presuming $N^+(u) \subseteq N^+(v)$  and $N^-(u) \subseteq N^-(v)$.
The implied isomorphism  here is, in some sense, trivial.
Each vertex other than $v$ maps to itself  and $v$ maps to $u_v$.
This occurs as there is a retraction $G \to G_{uv}$  where $u_v$ in $G_{uv}$ is relabelled as $v$.

An oriented graph or antisymmetric digraph $G$ is homomorphically full if and only if for all $u$ and $v$ that can be identified by an elementary homomorphism we have that $G_{uv}$ is isomorphic to a subgraph of $G$.
And so, similar to the case for graphs, if for an oriented graph or antisymmetric digraph $G$ every pair of vertices that can be identified by an elementary homomorphism is neighbourhood comparable, then $G$ is necessarily homomorphically full.
We will see that the converse of this statement holds for homomorphically full antisymmetric digraphs but not for homomorphically full oriented graphs.

Our remaining work proceeds as follows.
In Section \ref{sec:antiSym} we fully classify homomorphically full antisymmetric digraphs as those that are quasi-transitive and whose underlying graph is homomorphically full.
In doing so we provide a theorem for homomorphically full antisymmetric digraphs analogous to Theorem \ref{thm:main}.
These results imply that homomorphically full antisymmetric digraphs can be identified in polynomial time.
And also that one may decide in polynomial time if a graph is an underlying graph of a homomorphically full antisymmetric digraph.
In Section \ref{sec:oriented} we show for oriented graphs that neighbourhood comparability does not fully characterize  homomorphic fullness.
This work leads us to study, in Section \ref{sec:isFull}, the problem of deciding if an oriented graph is homomorphically full.
Though the analogous problems for graphs and antisymmetric digraphs are Polynomial, we show the problem to be GI-hard for oriented graphs.
We continue studying the complexity of problems related to homomorphic fullness of oriented graphs in Section \ref{sec:orientations}.
We find that deciding if a graph admits a homomorphically full orientation is NP-complete.
We conclude in Section \ref{sec:conclude} with discussion and further remarks.

Since the presence or absence of multiple edges in $\Gamma$ or $\Lambda$ does not matter in the definition of a 
homomorphism of $\Gamma$ to $\Lambda$, we consider only graphs  in which there are no multiple edges.
By contrast, the presence of loops matters. 
A homomorphism of a graph can map adjacent vertices to a vertex with a loop.
These statements remain true when we replace graphs with oriented graphs (under either interpretation).
To simplify matters, herein we assume that all graphs and oriented graphs are irreflexive.
We comment more on the nature of this problem for graphs that may have loops as part of our further remarks in Section \ref{sec:conclude}.

\section{Homomorphically Full Antisymmetric Digraphs}\label{sec:antiSym}

We begin our study of homomorphically full antisymmetric digraphs by examining their underlying graphs.

\begin{lemma}\label{lem:U(G)full}	
	If $G$ is a homomorphically full antisymmetric digraph, then its underlying graph $U(G)$ is homomorphically full.
\end{lemma}

\begin{proof}	
	Let $G$ be a homomorphically full antisymmetric digraph.
	Notice that each of $G$ and $U(G)$ have the same set of elementary homomorphisms.
	For $uv \notin A(G)$, if $G_{uv}$ is isomorphic to a subgraph of $G$ then $U(G_{uv})$ is necessarily isomorphic to a subgraph of $U(G)$.
	As $G$ is homomorphically full, for each $uv \notin A(G)$ we have that $G_{uv}$ is isomorphic to a subgraph of $G$.
	Therefore for each $uv \notin E(U(G))$ the graph $U(G)_{uv}$ is isomorphic to a subgraph of $U(G)$.
	Thus $U(G)$ is homomorphically full.
\end{proof}

We note that the converse of this lemma is false.  
As discussed in our introductory remarks, an undirected path on three vertices is homomorphically full, but the directed path on three vertices is not a  homomorphically full antisymmetric digraph.

Recall that a directed graph is \emph{quasi-transitive} when for each vertex $w$, there is complete adjacency between the in-neighbours and the out-neighbours of $w$.
Quasi-transitive directed graphs were first studied and subsequently fully classified by \cite{GH62}.
A graph is an underlying graph of a quasi-transitive digraph if and only if it is a comparability graph.
Further, every comparability graph is the underlying graph of some quasi-transitive digraph.

\begin{theorem}\label{thm:Defn1Char}
	An antisymmetric digraph is homomorphically full if and only if it is quasi-transitive and its underlying graph is homomorphically full.
	Further, every homomorphically full graph is the underlying graph of a homomorphically full antisymmetric digraph
\end{theorem}

\begin{proof}
	Let $G$ be an antisymmetric directed graph whose underlying graph is homomorphically full.
	
	If $G$ is not quasi-transitive, then there exists a vertex $w$ such that there is not complete adjacency between the in-neighbours and the out-neighbours of $w$.
	That is, there exists $u,v \in V(G)$ such that $uwv$ is an directed path of length $2$ (a $2$-dipath).
	The elementary homomorphism that identifies $u$ and $v$ is complete.
	As such $G_{uv}$ is a homomorphic image of $G$.
	However $G_{uv}$ is  not isomorphic to any subgraph of $G$; it contains a directed $2$-cycle.
	Therefore $G$ is not a homomorphically full antisymmetric digraph.
	
	Assume now that $G$ is quasi-transitive.
	Consider a pair of non-adjacent vertices $u,v \in V(G)$.
	Since $U(G)$ is homomorphically full, we have, without loss of generality, that $N(u) \subseteq N(v)$.
	Since $G$ is quasi-transitive, it follows that $N^+(u) \subseteq N^+(v)$  and $N^-(u) \subseteq N^-(v)$.
	Therefore $G_{uv} \cong G - u$.
	That is, $G_{uv}$ is isomorphic to a subgraph of $G$.
	Therefore $G$ is a homomorphically full antisymmetric digraph.

	Let $\Gamma$ be a homomorphically full  graph.
	By Theorem \ref{thm:main}, $\Gamma$ has no induced $P_4$, and hence is a cograph.
	Every cograph is a comparability graph.
	As proven by  \cite{GH62}, every comparability graph is the underlying graph of a quasi-transitive digraph.
\end{proof}

As with graphs, homomorphically full antisymmetric digraphs can be classified using neighbourhood comparability, homomorphic images and retracts.
We note, however, that such a classification does not exist for oriented graphs.
We explore this further in Section \ref{sec:oriented}.
\begin{theorem}\label{thm:Defn1Main}
	Let $G$ be an antisymmetric digraph.
	The following statements are equivalent:
	\begin{enumerate}[(a)]
		\item $G$ is a homomorphically full antisymmetric digraph.
		\item $G$ is  quasi-transitive and its underlying graph is homomorphically full.
		\item Every pair of non-adjacent vertices of $G$ are neighbourhood comparable.
		\item Every homomorphic image of $G$ is isomorphic to a retract of $G$.
		\item Every homomorphic image of $G$ is isomorphic to an induced subgraph of $G$.
	\end{enumerate}
\end{theorem}

\begin{proof}
	Let $G$ be a homomorphically full antisymmetric digraph.
	
	Assume $G$ is a homomorphically full antisymmetric digraph. 
	By Lemma \ref{lem:U(G)full}, we have that $U(G)$ is homomorphically full.
	By Theorem \ref{thm:Defn1Char} it follows that $G$ is quasi-transitive.
	That is, $(a) \Rightarrow (b)$.
	
	Assume $G$ is quasi-transitive and its underlying graph is homomorphically full.
	By Theorem \ref{thm:main} it follows that every pair of non-adjacent vertices of $U(G)$ is neighbourhood comparable.
	Since $G$ is quasi-transitive, it follows that $G$ has no induced $2$-dipath.
	Therefore every pair of non-adjacent vertices of $G$ are neighbourhood comparable.
	That is, $(b) \Rightarrow (c)$.
	
	Assume every pair of non-adjacent vertices of $G$ are neighbourhood comparable.
	Let $H$ be a homomorphic image of $G$.
	Let $\phi: G \to H$ be a complete homomorphism.
	Since $\phi$ is complete it suffices to assume it is an elementary homomorphism.
	Since $\phi$ is an elementary homomorphism we have that $H \cong G_{uv}$ for some pair of non-adjacent vertices $u$ and $v$.
	Since $u$ and $v$ are neighbourhood comparable, it follows without loss of generality that $G_{uv} \cong G - u$.
	Thus there exists a retraction $G \to G - u$.
	And so it follows that every homomorphic image of $G$ is isomorphic to a retract of $G$.
	That is, $(c) \Rightarrow (d)$.
	By definition, every retract of $G$ is isomorphic to an induced subgraph of $G$.
	Thus $(d) \Rightarrow (e)$.

	Finally, assume every homomorphic image of $G$ is isomorphic to an induced subgraph of $G$.
	From the definition of homomorphically full, it follows that $G$ is homomorphically full.
	That is, $(e) \Rightarrow (a)$.
	This completes the proof.
\end{proof}

Let $G$ be a homomorphically full  antisymmetric digraph and let $u$ and $v$ be a pair of non-adjacent vertices.
By Theorem \ref{thm:Defn1Main} we can assume, without loss of generality, that $G_{uv} \cong G - u$.
In some sense, this isomorphism is trivial -- vertices other than $u$ and $v$ may be mapped to themselves.
Vertex $u_v$ may be mapped to $v$.
By changing the label of $u_v$ to $v$ we arrive at a homomorphic image that is a subgraph of $G$.
The statement of Theorem \ref{thm:Defn1Main} suggests that for antisymmetric digraphs we may remove the word \emph{isomorphic} from within the definition of homomorphically full.
The same observation holds for homomorphically full graphs.
And in fact, the word \emph{isomorphism} does not appear in the original definition of homomorphically full given by \cite{B96}. 
In the following section, we will see that the existence of these trivial isomorphisms is not guaranteed for homomorphically full oriented graphs.

\section{Homomorphically Full Oriented Graphs}\label{sec:oriented}
We turn now to the interpretation of oriented graphs as arising from simple graphs and the subsequent definition of homomorphically full.
Consider the oriented graphs in Figure \ref{fig:counter}.
Since $v$ and $v^\prime$ are neighbourhood comparable we have  $ G-v^\prime \cong G_{vv^\prime}$.
Though $u$ and $u^\prime$ are not neighbourhood comparable, one can observe $G - v^\prime - xu^\prime \cong G_{uu^\prime}$.
%We notice here  $G_{uu^\prime}$ is isomorphic to a proper subgraph of $G$.
One can construct such an isomorphism by first noting that the sole vertex of degree $2$ in  $G-v^\prime - xu^\prime$ must map to the sole vertex of degree $2$ in $G_{uu^\prime}$.
As these are the only elementary homomorphisms of $G$, we conclude $G$ is homomorphically full.
The example shows that a theorem for homomorphically full oriented graphs akin to Theorems \ref{thm:main} and \ref{thm:Defn1Main} is not possible.
That is, this example shows that the property of homomorphic fullness is not necessarily related to neighbourhood comparability, nor is it necessarily related to the existence of retracts or homomorphisms to induced subgraphs.

\begin{figure}
	\[\includegraphics[scale=0.4]{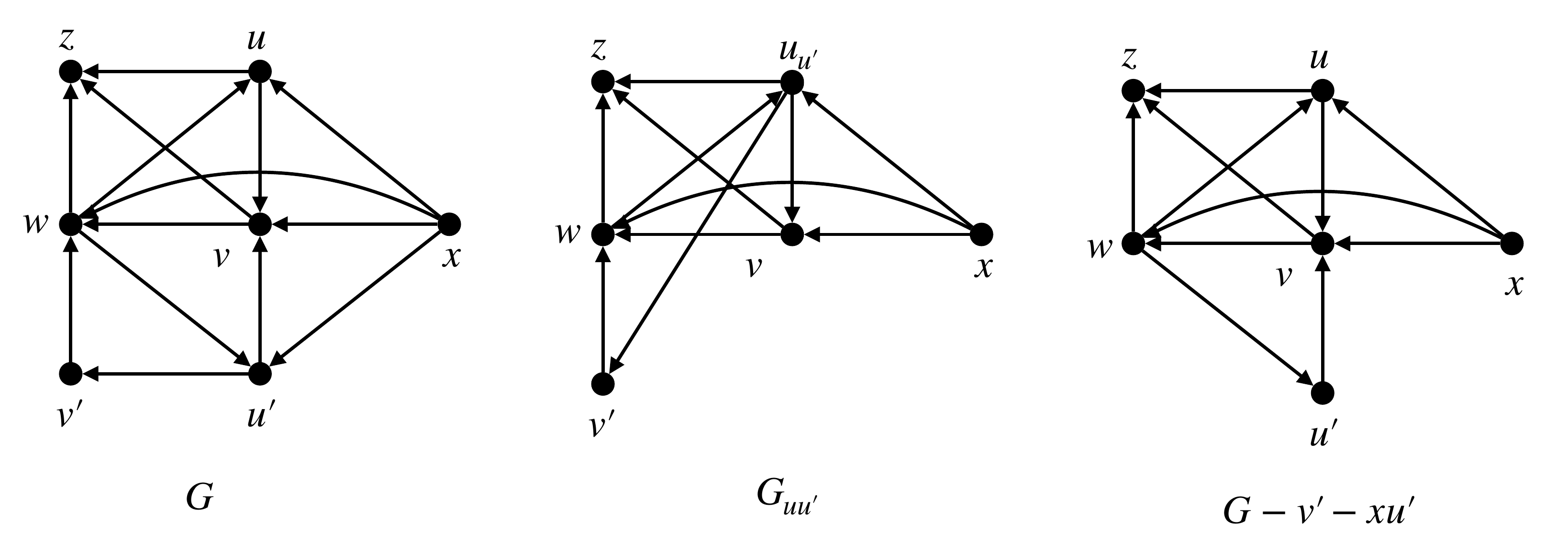}\]
	\caption{A homomorphically full oriented graph with a pair of vertices that are not neighbourhood comparable}
	\label{fig:counter}
\end{figure}

Recall that an \emph{oriented clique} is an oriented graph in which any two non-adjacent vertices are at directed distance $2$.
The name oriented clique arises from the fact an oriented colouring (i.e., a homomorphism) of such a graph assigns a different colour (i.e., image) for each vertex.
Observe that oriented cliques have no elementary homomorphisms.
From this observation the following two facts follow.
\begin{theorem}\label{thm:everyClique}
	Every oriented clique is a homomorphically full oriented graph.
\end{theorem}

\begin{theorem}\label{thm:isCore}
	Let $G$ be a homomorphically full oriented graph. The core of $G$ is an oriented clique.
\end{theorem}
\begin{proof}
	Let $G$ be a homomorphically full oriented graph and let $H$ be the oriented graph with the fewest vertices such that there is a homomorphism $\phi:G \to H$ that is onto with respect to $A(H)$. 
	Observe that $H$ is necessarily an oriented clique as if otherwise, $H$ has a pair of vertices $u$ and $v$ such that $G \to H_{uv}$.
	By the definition of homomorphically full, $H$ is isomorphic to a subgraph of $G$. 
	Therefore $H$ is isomorphic to the core of $G$.
\end{proof}

Let $G$ be an homomorphically full oriented graph.
As our definition of homomorphically full oriented graph restricts homomorphism to targets that are oriented graphs, an elementary homomorphism of $G$ necessarily identifies a pair of non-adjacent vertices that are not the ends of a $2$-dipath.
That is, an elementary homomorphism necessarily identifies vertices that are not at directed distance either $1$ or $2$.
We define the \emph{undirected closure of $G$}, denoted $cl(G)$, to be the graph formed from $G$ by adding an edge between any pair of vertices at the end of $2$-dipath and then considering all arcs as edges.
That is, we have $uv \in E(cl(G))$ when $u$ and $v$ are at directed distance at most $2$ in $G$.
Note that for any $u,v \in V(G)$ we have that there exists an elementary homomorphism $G \to G_{uv}$  if and only if  there exists an elementary homomorphism $cl(G) \to cl(G)_{uv}$.

\begin{lemma}\label{lem:subgraphClosure}
	Let $G$ be an oriented graph and let $u$ and $v$ be a pair of vertices that are neither adjacent nor at directed distance $2$.
	We have  $cl(G)_{uv} \subseteq cl(G_{uv})$.
\end{lemma}

\begin{proof}
	Let $G$ be an oriented graph.
	Let $u_v$ denote the vertex in both $cl(G_{uv})$ and $cl(G)_{uv}$ formed by identifying $u$ and $v$.
	As $cl(G_{uv})$ and $cl(G)_{uv}$ have the same vertex set, to show  $cl(G)_{uv} \subseteq cl(G_{uv})$	 it suffices to show that for all $xy \in E(cl(G)_{uv})$  we have $xy \in E(cl(G_{uv}))$.
	
	Consider first an arc $xy\in E(cl(G)_{uv})$ such that $x,y \neq u_v$.
	Since $xy\in E(cl(G)_{uv})$ and $x,y \neq u_v$, we have that $x$ and $y$ are adjacent in $U(G)$ or there is a $2$-dipath from $x$ to $y$ in $G$.
	Therefore $x$ and $y$ are adjacent in $cl(G)$ and so $xy \in  E(cl(G_{uv}))$.
	
	Assume without loss of generality that $x = u_v$.
	If one of $uy$ or $vy$ is an arc in $G$, then $u_vy$ is an arc in $G_{uv}$.
	Thus $xy$ is an edge in $cl(G_{uv}$.
	Otherwise, neither of $uy$ or $vy$ is an arc in $G$.
	Thus $xy$ as an edge in $cl(G)_{uv}$ could only have arisen from an edge of the form $uy$ or $vy$ in $cl(G)$.
	Therefore, there is a $2$-dipath in $G$ with one end at $y$ and the other at one of $u$ or $v$.
	Since $u$ and $v$ are not adjacent, this $2$-dipath exists in $G_{uv}$.
	Therefore $xy$ is an edge in $cl(G_{uv})$.
\end{proof}

Using Lemma \ref{lem:subgraphClosure} we find a connection between homomorphically full graphs and homomorphically full oriented graphs.

\begin{theorem}\label{thm:closure}
	If $G$ is a homomorphically full oriented graph, then $cl(G)$ is homomorphically full.
\end{theorem}

\begin{proof}	
	Let $G$ be an  homomorphically full oriented graph.
	Consider $cl(G)$ and $u,v \in V(G)$ such that $u$ and $v$ are not adjacent in $cl(G)$.
	Since $G$ is homomorphically full, it follows that $G_{uv}$ isomorphic to a subgraph of $G$.
	That is, there exists a subgraph $H$ of $G$ such that $G_{uv} \cong H$.
	Since $H$ is a subgraph of $G$, it follows that $cl(H)$ is a subgraph of $cl(G)$.
	And so by Lemma \ref{lem:subgraphClosure} we have
	
	\[ cl(G)_{uv} \subseteq cl(G_{uv}) \cong cl(H) \subseteq cl(G).\]
	
	Thus, $cl(G)_{uv}$ is isomorphic to a subgraph of $cl(G)$.
	Therefore $cl(G)$ is homomorphically full.
\end{proof}

\begin{corollary}
	A homomorphically full oriented graph has at most one nontrivial component.
\end{corollary}

\begin{proof}
	Let $G$ be a homomorphically full oriented graph.
	Notice that $U(G)$ and $cl(G)$ have the same number of nontrivial components.
	The result now follows from Theorems \ref{thm:main} and \ref{thm:closure}.
\end{proof}

Statement (e) in Theorem \ref{thm:main} implies that homomorphically full graphs admit a forbidden subgraph characterization.
Similarly, statement (b) in Theorem \ref{thm:Defn1Main} implies that homomorphically full antisymmetric digraphs admit a forbidden subgraph characterization.
We find this to not be the case for homomorphically full oriented graphs.

\begin{theorem}
	Every oriented graph is an induced subgraph of an homomorphically full oriented graph.
\end{theorem}

\begin{proof}
	By Lemma \ref{thm:everyClique}, it suffices to show that every oriented graph appears as an induced subgraph of some oriented clique.
	Let $G$ be an oriented graph with $n$ vertices.
	
	Let $B_n$ be the oriented complete bipartite graph with bipartition $(\{a_1, a_2, \ldots, $ $a_n\},$ $ \{b_1, b_2, \ldots, b_n\})$
	obtained by orienting the edge $a_ib_j$ from $a_i$ to $b_j$ if $i \leq j$ and from $b_j$ to $a_i$ if $i > j$.
	One can verify that $B_n$ is an oriented clique.
	
	Denote by $B^\prime_n$ the oriented clique obtained from the oriented clique $B_n$ by adding arcs such that the subgraph induced by $\{a_1, a_2, \ldots, a_n\}$ is isomorphic to $G$.
	
	By construction, the oriented graph $G$ is an induced subgraph of $B^\prime_n$.
	The result now follows by our previous remarks.
\end{proof}

As homomorphically full oriented graphs do not admit a forbidden subgraph orientation, one can wonder if there is an efficient recognition algorithm for homomorphically full oriented graphs.
We study this problem in the following section.

\section{Deciding if $G$ is Homomorphically Full}\label{sec:isFull}

We consider the problem of deciding if an oriented graph is in fact homomorphically full.
The corresponding decision problem for graphs and antisymmetric digraphs are Polynomial -- it is enough to compare the neighbourhoods of pairs of non-adjacent vertices.
However the example in Figure \ref{fig:hom} shows that such a procedure does not suffice for homomorphically full oriented graphs.
We see that $u$ and $u^\prime$ are not neighbourhood comparable, yet $G_{uu^\prime} \cong G-v$.
We use the example in Figure \ref{fig:hom} to show that the problem of deciding if an oriented graph is homomorphically full is GI-hard.
\begin{figure}
	\[\includegraphics[scale=0.5]{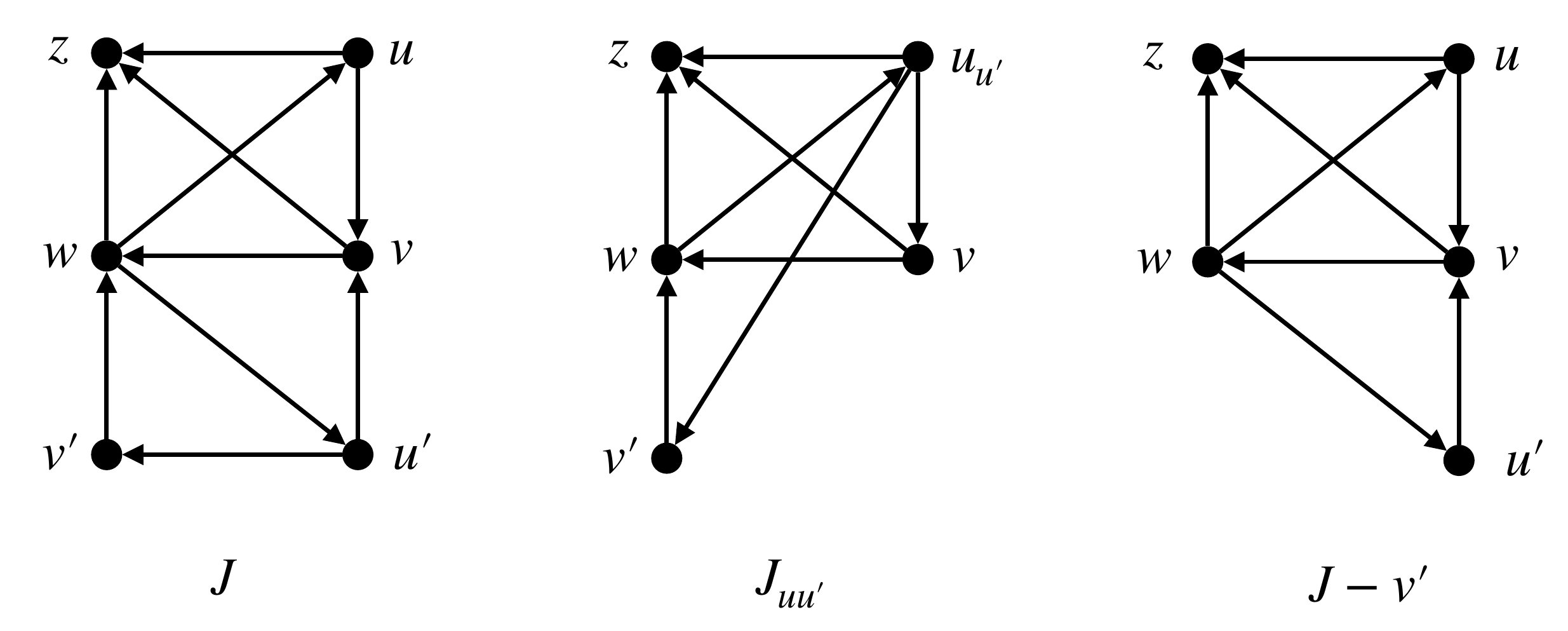}\]
	\caption{A homomorphically full oriented graph with a pair of vertices that are not neighbourhood comparable}
	\label{fig:hom}
\end{figure}
We do this by first showing that the problem of deciding if a pair of oriented cliques are isomorphic is GI-complete.
For these ends we define the following decision problems.

\medskip

\noindent\underline{HOMFULL}\\
\noindent \emph{Instance: }An oriented graph $G$.\\
\noindent \emph{Question:} Is $G$ homomorphically full?

\medskip
\noindent\underline{OCLIQUEISO}\\
\noindent \emph{Instance:} A pair of oriented cliques $G$ and $H$.\\
\noindent \emph{Question:} Is $G \cong H$?

\medskip
\noindent\underline{DAGISO}\\
\noindent \emph{Instance: }A pair of directed acyclic graphs $G$ and $H$\\
\noindent \emph{Question:} Is $G \cong H$?

We begin by showing OCLIQUEISO is GI-complete. 
For our reduction, we require the following result from Zemlyachenko et al.
\begin{theorem}\cite{Z85} \label{thm:dagiso}
	DAGISO is GI-complete.
\end{theorem}

Given an instance $G,H$ of DAGISO we construct oriented cliques $G^\star$ and $H^\star$ such that $G \cong H$ if and only if $G^\star \cong H^\star$.

We construct $G^\star$ using a pair of disjoint copies of $G$, say $G^L$ and $G^R$, and a fixed regular tournament $T$ with $2n+1$ vertices.
Let $V(G^L) = \{u^L_1, u^L_2,\dots, u^L_n \}$ and $V(G^R) = \{u^R_1, u^R_2,\dots, u^R_n \}$. 
Construct $G^\star$ by adding the following arcs to the oriented graph formed from the disjoint union of $G^L$, $G^R$ and $T$:
\begin{itemize}
	\item $u_i^Lu_i^R$ for all $1 \leq i \leq n$;
	\item $u_j^Ru_k^L$ for all $1 \leq j,k, \leq n, j\neq k$;
	\item $u^L_it$ for all $t \in V(T)$ and all $1 \leq i \leq n$; and
	\item $tu^R_i$ for all $t \in V(T)$ and all $1 \leq i \leq n$.
\end{itemize}

\begin{lemma}\label{lem:dagclique}
	Let $G$ and $H$ be directed acyclic graphs. We have  $G \cong H$ if and only if $G^\star \cong H^\star$.
\end{lemma}

\begin{proof}
	Let $G$ and $H$ be directed acyclic graphs with $n$ vertices. 
	Let $V(G) = \{u_1,u_2,\dots\, u_n\}$ and $V(H) = \{v_1,v_2,\dots, v_n\}$
	Let $\phi: G \to H$ be an isomorphism.
	Using $\phi$ we construct an isomorphism $\phi^\star: G^\star \to H^\star$ as follows:
	\begin{itemize}
		\item $\phi^\star(t) = t$ for all $t \in V(T)$;
		\item $\phi^\star(u_i^L) = \phi(u_i)^L$ for all $1 \leq i \leq n$; and
		\item $\phi^\star(u_i^R) = \phi(u_i)^R$ for all $1 \leq i \leq n$.
	\end{itemize}
	
	Assume now that $\beta^\star: G^\star \to H^\star$ is an isomorphism.
	It suffices to show that for all $u^L \in G^L$ we have  $\beta^\star(u^L) \in V(H^L)$.
	By observation $d^+(u^L) \geq 2n+2$.
	Therefore $d^+(\beta^\star(u^L)) \geq 2n+2$.
	In $H^\star$ for every $t \in V(T)$ we have $d^+(t) = 2n$.
	Further we have $d^+(v^R) \leq 2n-2$.
	Therefore $\beta^\star(u^L) \notin V(T) \cup V(H^R)$.
	It then follows that $\beta^\star(u^L) \in H_L$.
	Thus restricting $\beta^\star$ to $V(G^L)$ gives an isomorphism $G^L \to H^L$.
	Therefore $G \cong H$.
\end{proof}

\begin{theorem}\label{thm:ocliqueGIcomplete}
	OCLIQUEISO is GI-complete.
\end{theorem}

\begin{proof}
	The reduction is from DAGISO.
	Given an instance $G,H$ of DAGISO construct the instance $G^\star, H^\star$ of OCLIQUEISO. 
	Such a construction can be carried out in polynomial time.
	The result follows from Lemma \ref{lem:dagclique} and Theorem \ref{thm:dagiso}.
\end{proof}

Using Theorem \ref{thm:ocliqueGIcomplete}, we show HOMFULL is GI-hard.
Given an instance $G_1, G_2$ of OCLIQUEISO we construct an instance $\hat{G}$ of HOMFULL such that $G_1 \cong G_2$ if and only if $\hat{G}$ is homomorphically full.
We construct $\hat{G}$ from a copy of $J$ (as labelled as in Figure \ref{fig:hom}), a copy of $G_1$ and two copies of $G_2$, say $G_2$ and $G_2^\prime$, by adding a new vertex $q$ and the following arcs:
\begin{itemize}
	\item $qx$ for all $x \in V(J)$; and arcs $yq$ for $y \in V(G_1) \cup V(G_2) \cup V(G_2^{\prime})$
	\item $y_1y_2$, $y_2y_2^\prime$ and $y_2y_1$ for all $y_1\in  V(G_1)$, $y_2 \in V(G_2)$ and $y_2^\prime \in V(G_2^\prime)$.
	\item $wy_1$ for all $y_1 \in V(G_1)$.
	\item $vy_2$ for all $y_2 \in V(G_2)$.
	\item $uy_2^\prime$ for all $y_2^\prime \in V(G_2^\prime)$.
\end{itemize}

\begin{lemma}\label{lem:cliquefull}
	Let $G_1$ and $G_2$ be oriented cliques.
	We have $G_1 \cong G_2$ if and only if $\hat{G}$ is homomorphically full.
\end{lemma}

\begin{proof}
	Let $G_1$ and $G_2$ be oriented cliques.
	By construction $\hat{G}$ has two elementary homomorphisms: $G \to G_{uu^\prime}$ and $G \to G_{vv^\prime}$.
	As $v$ and $v^\prime$ are neighbourhood comparable, we have $G_{vv^\prime} \cong G - v^\prime$.
	Thus $\hat{G}$ is homomorphically full if and only if $\hat{G}_{uu^\prime}$ is isomorphic to some subgraph of $\hat{G}$.
	
	Let $\phi:G_2 \to G_1$ be an isomorphism
	We extend $\phi$ as follows such that $\phi^\star: \hat{G}_{uu^\prime} \to \hat{G} - v^\prime$ is an isomorphism.
	Let
	\begin{itemize}
		\item $\phi^\star(uu^\prime) = w$;
		\item $\phi^\star(v) = u$;
		\item $\phi^\star(w) = v$;
		\item $\phi^\star(v^\prime) = u^\prime$;
		\item $\phi^\star(y_2) = \phi(y_2)$ for all $y_2 \in V(G_2)$;
		\item $\phi^\star(y_1) = \phi^{-1}(y_1^\prime)$ for all $y_1 \in V(G_1)$; and
		\item $\phi^\star(y_2^\prime) = y_2$ for all $y_2^\prime$ in $V(G_2^\prime)$.
	\end{itemize}
	
	Assume now $\hat{G}$ is homomorphically full.
	Consider $\hat{C}$, the core of $\hat{G}$.
	By Theorem \ref{thm:isCore}, $\hat{C}$ is an oriented clique.
	By observation we produce a copy of $\hat{C}$ by removing $v^\prime$ and $u^\prime$ from $\hat{G}$.
	We see then that $\hat{C}$ is an oriented clique with $|V(G_1)| + 2|V(G_2)| + 4$ vertices.
	We also observe that neither $u^\prime$ nor $v^\prime$ are contained within a copy of $\hat{C}$.
	Consider now $\hat{G}r_{uu^\prime}$.
	Identifying a pair of vertices if $\hat{G}$ cannot change the core of $\hat{G}$ and so $\hat{C}$ is the core of $\hat{G}_{uu^\prime}$.
	This implies that $\hat{G}_{uu^\prime}$ has as a subgraph an oriented clique with $|V(G_1)| + 2|V(G_2)| + 4$ vertices.
	
	As $\hat{G}$ is homomorphically full, $\hat{G}_{uu^\prime}$ is isomorphic to some subgraph of $\hat{G} - s$ for some vertex $s \in V(\hat{G})$.	
	Since $\hat{G}- u^\prime$ has a vertex of degree $2$ (and $\hat{G}_{uu^\prime}$ does not), it cannot be that $s = u^\prime$.
	If $s \in V(G_1) \cup V(G_2) \cup V(G_2^\prime)  \cup \{q\}$, then $\hat{G}-s$ does not contain a copy of $\hat{C}$.
	Therefore $s = v^\prime$.
	Notice that $\hat{G}- v^\prime$ and $\hat{G}_{uu^\prime}$ have the same number of arcs and the same number of vertices.
	And so there is an isomorphism $\beta: \hat{G}_{uu^\prime} \to \hat{G} - v^\prime$.
	
	Each of $\hat{G}_{uu^\prime}$ and $\hat{G} - v^\prime$ has a single vertex of degree three.
	Therefore $\beta(v^\prime) = u^\prime$. 	
	In $\hat{G}_{uu^\prime}$, the vertex $v^\prime$ has a single out-neighbour: $w$.
	Therefore $\beta(w) = v$, as $v$ is the lone out-neighbour of the image of $v^\prime$ under $\beta$.
	In $\hat{G}_{uu^\prime}$, the vertex $v^\prime$ has a pair of in-neighbours: $q$ and $uu^\prime$.
	By considering the direction of the arc between $q$ and $uu^\prime$ we see $\beta(q) = q$ and $\beta(uu^\prime)=w$.
	Since $\beta(q) = q$ and $q$ has only four out-neighbours (three of which we have already considered), it follows that $\beta(v) = u$.	
	Since $\beta$ is an isomorphism, and all other vertices are accounted for, it must be that $\beta(t) \in V(G_1) \cup V(G_2) \cup V(G_2^\prime) $ for all $t \in V(G_1) \cup V(G_2) \cup V(G_2^\prime)$.
	As $\beta(w) = v$ it must be that $\beta(y) \in V(G_2)$ for all $y \in V(G_1)$.
	And so restricting $\beta$ to $V(G_1)$ gives an isomorphism $G_1 \to G_2$.
\end{proof}

\begin{theorem}
	HOMFULL is GI-hard.
\end{theorem}

\begin{proof}
	The reduction is from OCLIQUEISO.
	Given an instance $G_1,G_2$ of OCLIQUEISO construct the instance $\hat{G}$ of HOMFULL.
	Such a construction can be carried out in polynomial time.
	The result follows from Lemma \ref{lem:cliquefull} and Theorem \ref{thm:ocliqueGIcomplete}
\end{proof}

\section{Orientations of Homomorphically Full Graphs} \label{sec:orientations}
Theorem \ref{thm:Defn1Main} implies that one may decide in polynomial time if a graph is an underlying graph of some homomorphically full antisymmetric digraph. In this section we consider the related problem for homomorphically full oriented graphs.
We begin by noticing a relationship between homomorphically full graphs and homomorphically full oriented graphs.

\begin{theorem}\label{thm:GiveOrient}
	Every orientation of a homomorphically full graph is a homomorphically full oriented graph.
\end{theorem}

\begin{proof}
	Let $G$ be an orientation of a homomorphically full graph.
	We proceed by verifying that any pair of vertices that can be identified by an elementary homomorphism are neighbourhood comparable.
	Consider $u$ and $v$ such that $u$ and $v$ are either non-adjacent nor the ends of a $2$-dipath.
	If no such pair exists, then $G$ is an oriented clique,
	By Theorem \ref{thm:everyClique}, $G$ is homomorphically full.
	As $U(G)$ is homomorphically full, Theorem \ref{thm:main} implies, without loss of generality,  $N(u) \subseteq N(v)$.
	Since $u$ and $v$ are not the ends of a $2$-dipath we have \[N^+(u) \cap N^-(v) = N^+(v) \cap N^-(u) = \emptyset.\]
	Therefore $N^+(u) \subseteq N^+(v)$ and $N^-(u) \subseteq N^-(v)$.
	That is, $u$ and $v$ are neighbourhood comparable.
	This completes the proof.
\end{proof}

Recall from Theorem \ref{thm:Defn1Char} that the underlying graph of a homomorphically full antisymmetric graph is a homomorphically full graph.
From this we obtain the following corollary to Theorem \ref{thm:GiveOrient}.

\begin{corollary}
	Every homomorphically full antisymmetric graph is a homomorphically full oriented graph.
\end{corollary}

The converse to Theorem \ref{thm:GiveOrient} is false -- every oriented clique is homomorphically full.
The oriented clique given in Figure \ref{fig:planarClique} has an underlying graph that is not homomorphically full --- there exist pairs of non-adjacent vertices that are not neighbourhood comparable (see Theorem \ref{thm:main}).
However, the statement of Theorem \ref{thm:closure} implies that every homomorphically full oriented graph arises as a subgraph induced by the arc set of a mixed graph whose underlying graph is homomorphically full.
In the case of the oriented graph in Figure \ref{fig:planarClique}, this mixed graph is a partial orientation of a complete graph.

\begin{figure}
\[	\includegraphics[scale=0.5]{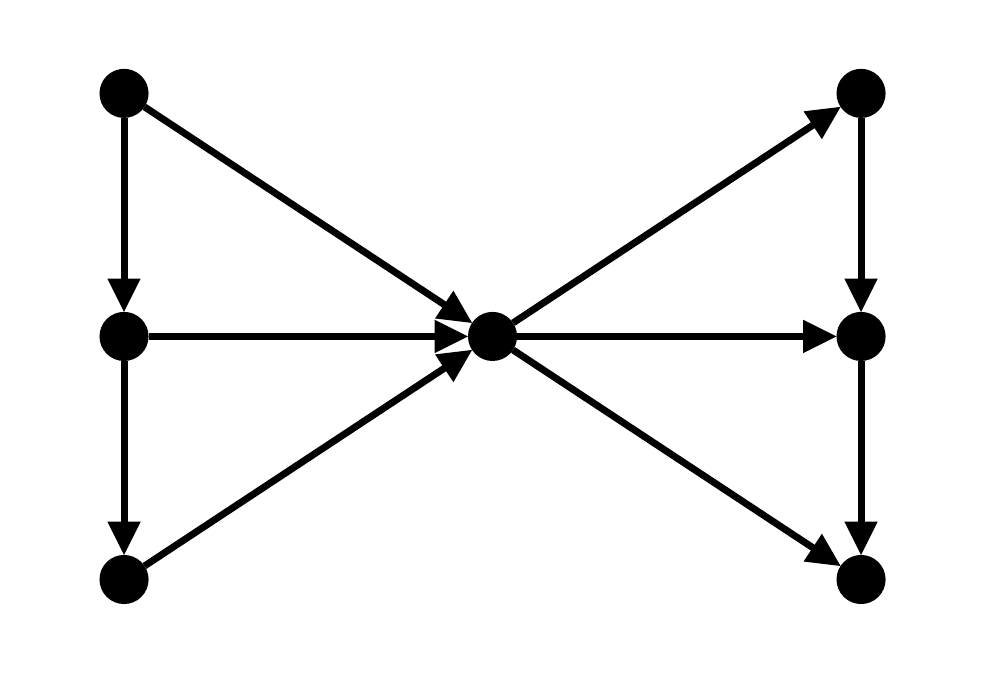}\]
	\caption{An oriented clique on $7$ vertices.}
	\label{fig:planarClique}
\end{figure}

We turn now to considering the problem of deciding if an undirected graph admits a homomorphically full orientation.

\medskip
\noindent\underline{FULLORIENT}\\
\noindent \emph{Instance: } A graph $\Gamma$.\\
\noindent \emph{Question:} Does $\Gamma$ admit an orientation that is  homomorphically full?

\medskip
We show FULLORIENT is NP-complete by way of reduction to the problem of deciding if a graph admits an orientation as an oriented clique.

\medskip
\noindent\underline{OCLIQUE}\\
\noindent \emph{Instance: } A simple graph $\Gamma$.\\
\noindent \emph{Question:} Does $\Gamma$ admit an orientation as an oriented clique?

\medskip

For our reduction, we require the following result of Kirgizov et al.
\begin{theorem}\cite{K16} \label{thm:oCliqueNPC}
	OCLIQUE is NP-complete.
\end{theorem}

Given an instance $\Gamma$, of OCLIQUE we construct an instance $\tilde{\Gamma}$ of FULLORIENT such that $\Gamma$ admits an orientation as an oriented clique if and only if $\tilde{\Gamma}$ admits an orientation that is homomorphically full.
Let $V(\Gamma) = \{v_1, v_2, \ldots, v_n\}$.
Let $\Lambda$ be a complete graph on $n+2$ vertices with vertex set $V(\Lambda) = \{v_1^\prime, v_2^\prime, \ldots, v_n^\prime, s, t\}$.
We construct $\tilde{\Gamma}$ from the disjoint union of $\Gamma$ and a complete graph by adding the following edges
\begin{itemize}
	\item $v_iv_i^\prime$, for all $1 \leq i \leq n$, and
	\item $tv_i$ for each $1 \leq i \leq n$.
\end{itemize}

\begin{lemma}\label{lem:orientfull}
	Let $\Gamma$ be a graph.
	The graph $\Gamma$ admits an orientation as an oriented clique if and only if $\tilde{\Gamma}$ admits an orientation that is homomorphically full.
\end{lemma}

\begin{proof}
	Let $\Gamma$ be a graph.
	Notice that no two vertices of $\tilde{\Gamma}$  are neighbourhood comparable.
	Therefore $\tilde{\Gamma}$ has an orientation that is a homomorphically full oriented graph if and only if it has an orientation that is an oriented clique.
	
	Suppose $\Gamma$ has an orientation as an oriented clique, $G$.
	Extend this orientation to  $\tilde{\Gamma}$ by orienting the edge between $v_i$ and $v_i^\prime$ from $v_i$ to $v_i^\prime$, $1 \leq i \leq n$; orienting $\Lambda$ to be a transitive tournament, $H$, in which $s$ has in-degree 0 and $t$ has out-degree $0$; and orienting the edges between $t$ and $v_1, v_2, \ldots, v_n$ from $t$ to $v_i$, $1 \leq i \leq n$.
	
	Since $G$ and $H$ are oriented cliques, it remains to verify that there is a directed path of length at most 2 between each vertex of $G$ and each vertex of $H - t$.  For $x \in V(H) - \{t\}$ and $1 \leq i \leq n$ there is a $2$-dipath $x,t,v_i$.  Thus this orientation of  $\tilde{\Gamma}$ is an oriented clique, and hence is a homomorphically full oriented graph.
	
	Now suppose $\tilde{\Gamma}$ has an orientation $\tilde{G}$ that is a homomorphically full oriented graph. 
	As noted above, this orientation $\tilde{G}$  is an oriented clique.
	Let $G$ and $H$ respectively be the subgraphs of $\tilde{G}$  induced by the vertex sets of $\Gamma$ and $\Lambda$.
	Since there is no path of length 2 in $\tilde{\Gamma}$ joining a vertex of $\Gamma$ and a vertex of $\Lambda$, no $2$-dipath joining vertices of $G$ contains a vertex of $H$.
	Therefore $G$ is an oriented clique.
	And so $\Gamma$ admits an orientation as an oriented clique.
\end{proof}

\begin{theorem}
	FULLORIENT is NP-complete. 
\end{theorem}

\begin{proof}
	The transformation is from OCLIQUE.  
	Given an instance $\Gamma$ of OCLIQUE construct the instance $\tilde{\Gamma}$ of FULLORIENT.
	Such a construction can be carried out in polynomial time.
	The result follows from Lemma \ref{lem:orientfull} and Theorem \ref{thm:oCliqueNPC}
\end{proof}

\section{Further Remarks and Future Work} \label{sec:conclude}
The study of homomorphisms oriented graphs often goes hand-in-hand with that of $2$-edge-coloured graphs.
The literature is brimming of examples where similar looking results and methods appear for these objects.
For example, see results for homomorphisms of oriented and $2$-edge coloured planar graphs by \cite{AM98} and by \cite{RASO94}.

Emulating the results in Section \ref{sec:antiSym} for $2$-edge-coloured graphs would require extending the notion of quasi-transitivity to $2$-edge-coloured graphs.
As shown by  \cite{DM21},  the classification of those graphs that admit a $2$-edge-colouring with a property analogous to quasi-transitivity results in a classification that is much less natural than that for oriented graphs.
As such we expect discovering the analogous results for $2$-edge-coloured graphs  to require significant work.
Similarly, emulating the results in Section \ref{sec:orientations} for $2$-edge-coloured graphs will require significant work.
It it unknown if a result analogous to Theorem \ref{thm:oCliqueNPC} holds for a $2$-edge-coloured version of the problem.

The result of Theorem \ref{thm:main} assumes that both the input graph and the target graph of the homomorphism are irreflexive.
As shown by \cite{H13}, statements in this theorem corresponding to (a) through (d) are also equivalent for reflexive graphs, but not for reflexive digraphs.
Note that the definition of neighbourhood comparability differs slightly to account for the loops, as there is no assumption that adjacent vertices must have the different images under a homomorphism.
In the case of reflexive graphs, statements (e) and (f) respectively become ``$G$ has none of $C_4$, $P_4$, $2K_2$ as an induced subgraph, i.e., $G$ is a threshold graph'' and ``$\overline{G}$ is the comparability graph of a threshold order''.
The homomorphically full reflexive semi-complete digraphs are characterized in the same paper.
The complete characterization of homomorphically full reflexive and irreflexive digraphs remains open.

\bibliographystyle{abbrvnat}
\bibliography{references.bib}

\end{document}